\documentclass[11pt]{article} \usepackage[normal]{optional}
\usepackage{amsmath}
\usepackage{amssymb}
\usepackage{algorithmic}
\usepackage{algorithm}
\usepackage{graphicx}
\usepackage{ifpdf}
\usepackage{cite}
\usepackage{comment}
\usepackage{subfigure}
\usepackage{todonotes} 

\newcommand{\figuresize}{\small}

\opt{normal,submission}{\usepackage{commands-tam,numbering-tam}}

\usepackage{amsfonts}
\opt{normal,submission}{\usepackage{amsthm}}
\usepackage{fullpage}

\newcommand{\U}{\Upsilon}
\newcommand{\Ua}{\Upsilon_\alpha}
\newcommand{\Ub}{\Upsilon_\beta}

\newcommand{\dN}{\mathsf{N}}
\newcommand{\dS}{\mathsf{S}}
\newcommand{\dE}{\mathsf{E}}
\newcommand{\dW}{\mathsf{W}}
\newcommand{\dall}{\{\dN,\dS,\dE,\dW\}}

\newcommand{\ignore}[1]{}

\newcommand{\vv}{\vec{v}}

\newcommand{\hUPV}{{\normalfont \textsf{hUPV}_1}}
\newcommand{\sUPV}{{\normalfont \textsf{sUPV}_1}}

\ifpdf
  \usepackage[pdftex]{epsfig}
  \usepackage{pageno}
  \usepackage[pdfstartview={XYZ null null 1.00},pdfpagemode=None,colorlinks=true,citecolor=blue,linkcolor=blue]{hyperref}
\else
    \usepackage[dvips]{epsfig}
    \usepackage[dvips,colorlinks=true,citecolor=blue,linkcolor=blue]{hyperref}
\fi

\begin{document}

\title{Producibility in hierarchical self-assembly}

\author{
    David Doty\thanks{California Institute of Technology, Pasadena, CA, USA, {\tt ddoty@caltech.edu}. The author was supported by the Molecular Programming Project under NSF grant 0832824 and by NSF grants CCF-1219274 and CCF-1162589, and by a Computing Innovation Fellowship under NSF grant 1019343.}
}
\date{}

\maketitle
\begin{abstract}
  Three results are shown on producibility in the hierarchical model of tile self-assembly.

  It is shown that a simple greedy polynomial-time strategy decides whether an assembly $\alpha$ is producible.
  The algorithm can be optimized to use $O(|\alpha| \log^2 |\alpha|)$ time.

  Cannon, Demaine, Demaine, Eisenstat, Patitz, Schweller, Summers, and Winslow~\cite{TwoHandsBetterThanOne} showed that the problem of deciding if an assembly $\alpha$ is the unique producible terminal assembly of a tile system $\calT$ can be solved in $O(|\alpha|^2 |\calT| + |\alpha| |\calT|^2)$ time for the special case of noncooperative ``temperature 1'' systems.
  It is shown that this can be improved to $O(|\alpha| |\calT| \log |\calT|)$ time.

  Finally, it is shown that if two assemblies are producible, and if they can be overlapped consistently -- i.e., if the positions that they share have the same tile type in each assembly -- then their union is also producible.
\end{abstract}


\thispagestyle{empty}\newpage\setcounter{page}{1}
\vspace{-1.1cm}
\section{Introduction}
\label{sec-intro}

\subsection{Background of the field}
Winfree's abstract Tile Assembly Model (aTAM)~\cite{Winfree98simulationsof} is a model of crystal growth through cooperative binding of square-like monomers called \emph{tiles}, implemented experimentally (for the current time) by DNA~\cite{WinLiuWenSee98, BarSchRotWin09}.
In particular, it models the potentially algorithmic capabilities of tiles that can be designed to bind if and only if the total strength of attachment (summed over all binding sites, called \emph{glues} on the tile) is at least a parameter $\tau$, sometimes called the \emph{temperature}.
In particular, when the glue strengths are integers and $\tau=2$, this implies that two strength 1 glues must cooperate to bind the tile to a growing assembly.
Two assumptions are key: 1) growth starts from a single specially designated \emph{seed} tile type, and 2) only individual tiles bind to an assembly, never larger assemblies consisting of more than one tile type.
We will refer to this model as the \emph{seeded aTAM}.
While violations of these properties are often viewed as errors in implementation of the seeded aTAM~\cite{SchWin07, SchWin09}, relaxing the assumptions results in a different model with its own programmable abilities.
In the \emph{hierarchical} (a.k.a. \emph{multiple tile}~\cite{AGKS05}, \emph{polyomino}~\cite{Winfree06, Luhrs10}, \emph{two-handed}~\cite{TwoHandsBetterThanOne, DotPatReiSchSum10, TwoHandedNotIntrinsicallyUniversalconf}) \emph{aTAM}, there is no seed tile, and an assembly is considered producible so long as two producible assemblies are able to attach to each other with strength at least $\tau$, with all individual tiles being considered as ``base case'' producible assemblies.
In either model, an assembly is considered \emph{terminal} if nothing can attach to it; viewing self-assembly as a computation, terminal assembly(ies) are often interpreted to be the output.
See~\cite{DotCACM,PatitzSurvey12} for an introduction to recent work in these models.

The hierarchical aTAM has attracted considerable recent attention.
It is $\coNP$-complete to decide whether an assembly is the unique terminal assembly produced by a hierarchical tile system~\cite{TwoHandsBetterThanOne}.
There are infinite shapes that can be assembled in the hierarchical aTAM but not the seeded aTAM, and vice versa, and there are finite shapes requiring strictly more tile types to assemble in the seeded aTAM than the hierarchical aTAM, and vice versa~\cite{TwoHandsBetterThanOne}.
Despite this incomparability between the models for exact assembly of shapes, with a small blowup in scale, any seeded tile system can be simulated by a hierarchical tile system~\cite{TwoHandsBetterThanOne}, improving upon an earlier scheme that worked for restricted classes of seeded tile systems~\cite{Luhrs10}.
However, the hierarchical aTAM is not able to simulate itself from a single set of tile types, i.e., it is not \emph{intrinsically universal}~\cite{TwoHandedNotIntrinsicallyUniversalconf}, unlike the seeded aTAM~\cite{TAMIU}.
It is possible to assemble an $n \times n$ square in a hierarchical tile system with $O(\log n)$ tile types that exhibits a very strong form of fault-tolerance in the face of spurious growth via strength 1 bonds~\cite{DotPatReiSchSum10}.
The parallelism of the hierarchical aTAM suggests the possibility that it can assemble shapes faster than the seeded aTAM, but it cannot for a wide class of tile systems~\cite{CheDot12}.

Interesting variants of the hierarchical aTAM introduce other assumptions to the model.
The \emph{multiple tile} model retains a seed tile and places a bound on the size of assemblies attaching to it~\cite{AGKS05}.
Under this model, it is possible to modify a seeded tile system to be \emph{self-healing}, that is, it correctly regrows when parts of itself are removed, even if the attaching assemblies that refill the removed gaps are grown without the seed~\cite{Winfree06}.
The model of \emph{staged assembly} allows multiple test tubes to undergo independent growth, with excess incomplete assemblies washed away (e.g. purified based on size) and then mixed, with assemblies from each tube combining via hierarchical attachment~\cite{DDFIRSS07, DemEisIshWinOneDimStagedjournal}.
The \emph{RNase enzyme} model~\cite{RNaseSODA2010, DemPatSchSum2010RNase, PatSumIdentifyingJournal} assumes some tile types to be made of RNA, which can be digested by an enzyme called RNase, leaving only the DNA tiles remaining, and possibly disconnecting what was previously a single RNA/DNA assembly into multiple DNA assemblies that can combine via hierarchical attachment.
Introducing negative glue strengths into the hierarchical aTAM allows for ``fuel-efficient'' computation~\cite{SchSheSODAFuel}.

\subsection{Contributions of this paper}
We show three results related to producibility in the hierarchical aTAM.
\begin{enumerate}
  \item
  In the seeded aTAM, there is an obvious linear-time algorithm to test whether assembly $\alpha$ is producible by a tile system: starting from the seed, try to attach tiles until $\alpha$ is complete or no more attachments are possible.
  We show that in the hierarchical aTAM, a similar greedy strategy correctly identifies whether a given assembly is producible, though it is more involved to prove that it is correct.
  The idea is to start with all tiles in place as they appear in $\alpha$, but with no bonds, and then to greedily bind attachable assemblies until $\alpha$ is assembled.
  The algorithm can be optimized to use $O(|\alpha| \log^2 |\alpha|)$ time.
  It is not obvious that this works, since it is conceivable that assemblies must attach in a certain order for $\alpha$ to form, but the greedy strategy may pick another order and hit a dead-end in which no assemblies can attach.
  This is shown in Section~\ref{sec-algorithm}.

  \item
  We show that there is a faster algorithm for the temperature 1 Unique Production Verification (UPV) problem studied by Cannon, Demaine, Demaine, Eisenstat, Patitz, Schweller, Summers, and Winslow~\cite{TwoHandsBetterThanOne}, which is the problem of determining whether assembly $\alpha$ is the unique producible terminal assembly of tile system $\calT$, where $\calT$ has temperature 1, meaning that all positive glue strength are sufficiently strong to attach any two assemblies.
  They give an algorithm that runs in $O(|\alpha|^2 |\calT| + |\alpha| |\calT|^2)$ time.
  We show how to improve this running time.
  Cannon et al.~proved their result by using an $O(|\alpha|^2 + |\alpha| |\calT|)$ time algorithm for UPV that works in the seeded aTAM~\cite{ACGHKMR02}, and then reduced the hierarchical temperature-1 UPV problem to $|\calT|$ instances of the seeded UPV problem.
  We improve this result by showing that a faster $O(|\alpha| \log |\calT|)$ time algorithm for the seeded UPV problem exists for the special case of temperature 1, and then we apply the technique of Cannon et al.~relating the hierarchical problem to the seeded problem to improve the running time of the hierarchical algorithm to $O(|\alpha| |\calT| \log |\calT|)$.
  This is shown in Section~\ref{sec-upv-t1}.

  \item
  We show that if two assemblies $\alpha$ and $\beta$ are producible in the hierarchical model, and if they can be overlapped consistently (i.e., if the positions that they share have the same tile type in each assembly), then their union $\alpha \cup \beta$ is producible.
  This is trivially true in the seeded model, but it requires more care to prove in the hierarchical model.
  It is conceivable \emph{a priori} that although $\beta$ is producible, $\beta$ must assemble $\alpha \cap \beta$ in some order that is inconsistent with how $\alpha$ assembles $\alpha \cap \beta$.
  This is shown in Section~\ref{sec-union}.
\end{enumerate}

\vspace{-0.2in}
\section{Informal definition of the abstract tile assembly model}
\label{sec-tam-informal}

This section gives a brief informal sketch of the seeded and hierarchical variants of the abstract Tile Assembly Model (aTAM).
See Section \ref{sec-tam-formal} for a formal definition of the aTAM.

A \emph{tile type} is a unit square with four sides, each consisting of a \emph{glue label} (often represented as a finite string) and a nonnegative integer \emph{strength}.
We assume a finite set $T$ of tile types, but an infinite number of copies of each tile type, each copy referred to as a \emph{tile}.
If a glue has strength 0, we say it is \emph{null}, and if a positive-strength glue facing some direction does not appear on some tile type in the opposite direction, we say it is \emph{functionally null}.
We assume that all tile sets in this paper contain no functionally null glues.
An \emph{assembly}
is a positioning of tiles on the integer lattice $\Z^2$; i.e., a partial function $\alpha:\Z^2 \dashrightarrow T$. 
We write $|\alpha|$ to denote $|\dom \alpha|$.
Write $\alpha \sqsubseteq \beta$ to denote that $\alpha$ is a \emph{subassembly} of $\beta$, which means that $\dom\alpha \subseteq \dom\beta$ and $\alpha(p)=\beta(p)$ for all points $p\in\dom\alpha$.
In this case, say that $\beta$ is a \emph{superassembly} of $\alpha$.
We abuse notation and take a tile type $t$ to be equivalent to the single-tile assembly containing only $t$ (at the origin if not otherwise specified).
Two adjacent tiles in an assembly \emph{interact} if the glue labels on their abutting sides are equal and have positive strength. 
Each assembly induces a \emph{binding graph}, a grid graph whose vertices are tiles, with an edge between two tiles if they interact.
The assembly is \emph{$\tau$-stable} if every cut of its binding graph has strength at least $\tau$, where the weight of an edge is the strength of the glue it represents.

A \emph{seeded tile assembly system} (seeded TAS) is a triple $\calT = (T,\sigma,\tau)$, where $T$ is a finite set of tile types, $\sigma:\Z^2 \dashrightarrow T$ is a finite, $\tau$-stable \emph{seed assembly}, 
and $\tau$ is the \emph{temperature}.
If $\calT$ has a single seed tile $s\in T$ (i.e., $\sigma(0,0)=s$ for some $s \in T$ and is undefined elsewhere), then we write $\calT=(T,s,\tau).$
Let $|\calT|$ denote $|T|$.
An assembly $\alpha$ is \emph{producible} if either $\alpha = \sigma$ or if $\beta$ is a producible assembly and $\alpha$ can be obtained from $\beta$ by the stable binding of a single tile.
In this case write $\beta\to_1 \alpha$ ($\alpha$ is producible from $\beta$ by the attachment of one tile), and write $\beta\to \alpha$ if $\beta \to_1^* \alpha$ ($\alpha$ is producible from $\beta$ by the attachment of zero or more tiles).
An assembly is \emph{terminal} if no tile can be $\tau$-stably attached to it.

A \emph{hierarchical tile assembly system} (hierarchical TAS) is a pair $\calT = (T,\tau)$, where $T$ is a finite set of tile types and $\tau \in \N$ is the temperature.
An assembly is \emph{producible} if either it is a single tile from $T$, or it is the $\tau$-stable result of translating two producible assemblies without overlap.
Therefore, if an assembly $\alpha$ is producible, then it is produced via an \emph{assembly tree}, a full binary tree whose root is labeled with $\alpha$, whose $|\alpha|$ leaves are labeled with tile types, and each internal node is a producible assembly formed by the stable attachment of its two child assemblies.
An assembly $\alpha$ is \emph{terminal} if for every producible assembly $\beta$, $\alpha$ and $\beta$ cannot be $\tau$-stably attached.
If $\alpha$ can grow into $\beta$ by the attachment of zero or more assemblies, then we write $\alpha \to \beta$.

Note that our definitions imply only finite assemblies are producible.
In either the seeded or hierarchical model, let $\prodasm{\calT}$ be the set of producible assemblies of $\calT$, and let $\termasm{\calT} \subseteq \prodasm{\calT}$ be the set of producible, terminal assemblies of $\calT$.
A TAS $\calT$ is \emph{directed} (a.k.a., \emph{deterministic}, \emph{confluent}) if $|\termasm{\calT}| = 1$.
If $\calT$ is directed with unique producible terminal assembly $\alpha$, we say that $\calT$ \emph{uniquely produces} $\alpha$.
It is easy to check that in the seeded aTAM, $\calT$ uniquely produces $\alpha$ if and only if every producible assembly $\beta \sqsubseteq \alpha$.
In the hierarchical model, a similar condition holds, although it is more complex since hierarchical assemblies, unlike seeded assemblies, do not have a ``canonical translation'' defined by the seed.
$\calT$ uniquely produces $\alpha$ if and only if for every producible assembly $\beta$, there is a translation $\beta'$ of $\beta$ such that $\beta' \sqsubseteq \alpha$.
In particular, if there is a producible assembly $\beta \neq \alpha$ such that $\dom \alpha = \dom \beta$, then $\alpha$ is not uniquely produced.
Since $\dom \beta = \dom \alpha$, every nonzero translation of $\beta$ has some tiled position outside of $\dom \alpha$, whence no such translation can be a subassembly of $\alpha$, implying $\alpha$ is not uniquely produced.


\vspace{-0.2in}
\section{Polynomial-time verification of production}
\label{sec-algorithm}

Let $S$ be a finite set.
A \emph{partition} of $S$ is a collection $\calC = \{ C_1, \ldots, C_k \} \subseteq \calP(S)$ such that $\bigcup_{i=1}^k C_i = S$ and for all $i \neq j$, $C_i \cap C_j = \emptyset$.
A \emph{hierarchical division} of $S$ is a full binary tree $\U$ (a tree in which every internal node has exactly two children) whose nodes represent subsets of $S$, such that the root of $\U$ represents $S$, the $|S|$ leaves of $\U$ represent the singleton sets $\{x\}$ for each $x \in S$, and each internal node has the property that its set is the (disjoint) union of its two childrens' sets.

\begin{lemma}\label{lem-partition-tree-partition}
Let $S$ be a finite set with $|S| \geq 2$.
Let $\U$ be any hierarchical division of $S$, and let $\calC$ be any partition of $S$ other than $\{S\}$.
Then there exist $C_1,C_2 \in \calC$ with $C_1 \neq C_2$, and there exist $C'_1 \subseteq C_1$ and $C'_2 \subseteq C_2$, such that $C'_1$ and $C'_2$ are siblings in $\U$.
\end{lemma}

\begin{proof}
First, label each leaf $\{x\}$ of $\U$ with the unique element $C_i \in \calC$ such that $x \in C_i$.
Next, iteratively label internal nodes according to the following rule: while there exist two children of a node $u$ that have the same label, assign that label to $u$.
Notice that this rule preserves the invariant that each labeled node $u$ (representing a subset of $S$) is a subset of the set its label represents.
Continue until no node has two identically-labeled children.
$\calC$ contains only proper subsets of $S$, so the root (which is the set $S$) cannot be contained in any of them, implying the root will remain unlabeled.
Follow any path starting at the root, always following an unlabeled child, until both children of the current internal node are labeled.
(The path may vacuously end at the root.)
Such a node is well-defined since at least all leaves are labeled.
By the stopping condition stated previously, these children must be labeled differently.
The children are the witnesses $C'_1$ and $C'_2$, with their labels having the values $C_1$ and $C_2$, testifying to the truth of the lemma.
\end{proof}

Lemma \ref{lem-partition-tree-partition} will be useful when we view $\Upsilon$ as an assembly tree for some producible assembly $\alpha$, and we view $\calC$ as a partially completed attempt to construct another assembly tree for $\alpha$, where each element of $\calC$ is a subassembly that has been produced so far.

When we say ``by monotonicity'', this refers to the fact that glue strengths are nonnegative, which implies that if two assemblies $\alpha$ and $\beta$ can attach, the addition of more tiles to either $\alpha$ or $\beta$ cannot \emph{prevent} this binding, so long as the additional tiles do not overlap the other assembly.

\begin{algorithm}
\caption{\textsc{Is-Producible-Assembly}$(\alpha,\tau)$}
\label{alg-slow}
\begin{algorithmic}[1]
  \STATE {\bf input:} assembly $\alpha$ and temperature $\tau$
  \STATE $\calC \leftarrow \setr{ \{v\} }{ v \in \dom \alpha }$\ \ \ // \emph{(positions defining) subassemblies of $\alpha$; initially individual (positions of) tiles}
  \WHILE{$|\calC| > 1$} \label{start-while}
    \IF{there exist $C_i,C_j\in\calC$ with glues between $C_i$ and $C_j$ of total strength at least $\tau$} \label{get-glues}
      \STATE $\calC \leftarrow (\calC \setminus \{C_i,C_j\}) \cup \{ C_i \cup C_j \}$ \label{modify-C}
    \ELSE
      \PRINT{``$\alpha$ is not producible''} and {\bf exit}
    \ENDIF
  \ENDWHILE
  \PRINT{``$\alpha$ is producible''}
\end{algorithmic}
\end{algorithm}

We want to solve the following problem: given an assembly $\alpha$ and temperature $\tau$, is $\alpha$ producible in the hierarchical aTAM at temperature $\tau$?\footnote{We do not need to give the tile set $T$ as input because the tiles in $\alpha$ implicitly define a tile set, and the presence of extra tile types in $T$ that do not appear in $\alpha$ cannot affect its producibility.}
The algorithm \textsc{Is-Producible-Assembly} (Algorithm~\ref{alg-slow}) solves this problem.

\begin{theorem}
  There is an $O(|\alpha| \log^2 |\alpha|)$ algorithm deciding whether an assembly $\alpha$ is producible at temperature $\tau$ in the hierarchical aTAM.
\end{theorem}

\begin{proof}
\noindent {\bf Correctness:}
\textsc{Is-Producible-Assembly} works by building up the initially edge-free graph with the tiles of $\alpha$ as its nodes (the algorithm stores the nodes as points in $\Z^2$, but $\alpha$ would be used in step \ref{get-glues} to get the glues and strengths between tiles at adjacent positions), stopping when the graph becomes connected.
The order in which connected components (implicitly representing assemblies) are removed from and added to $\calC$ implicitly defines a particular assembly tree with $\alpha$ at the root (for every $C_1,C_2$ processed in line \ref{modify-C}, the assembly $\alpha \upharpoonright (C_1 \cup C_2)$ is a parent of $\alpha \upharpoonright C_1$ and $\alpha \upharpoonright C_2$ in the assembly tree).
Therefore, if the algorithm reports that $\alpha$ is producible, then it is.
Conversely, suppose that $\alpha$ is producible via assembly tree $\Upsilon$. 
Let $\calC = \{C_1,\ldots,C_k\}$ be the set of assemblies at some iteration of the loop at line \ref{start-while}.
It suffices to show that some pair of assemblies $C_i$ and $C_j$ are connected by glues with strength at least $\tau$.
By Lemma \ref{lem-partition-tree-partition}, there exist $C_i$ and $C_j$ with subsets $C'_i \subseteq C_i$ and $C'_j \subseteq C_j$ such that $C'_i$ and $C'_j$ are sibling nodes in $\Upsilon$.
Because they are siblings, the glues between $C'_i$ and $C'_j$ have strength at least $\tau$.
By monotonicity these glues suffice to bind $C_i$ to $C_j$,
so \textsc{Is-Producible-Assembly} is correct.

\noindent{\bf Running time:} Let $n=|\alpha|$.
The running time of the \textsc{Is-Producible-Assembly} (Algorithm~\ref{alg-slow}) is polynomial in $n$, but the algorithm can be optimized to improve the running time to $O(n \log^2 n)$ by careful choice of data structures.
\textsc{Is-Producible-Assembly-Fast} (Algorithm~\ref{alg-fast}) shows pseudo-code for this optimized implementation, which we now describe.
Let $n = |\alpha|$. 
Instead of searching over all pairs of assemblies, only search those pairs of assemblies that are adjacent.
This number is $O(n)$ since a grid graph has degree at most 4 (hence $O(n)$ edges) and the number of edges in the full grid graph of $\alpha$ is an upper bound on the number of adjacent assemblies at any time.
This can be encoded in a dynamically changing graph $G_c$ whose nodes are the current set of assemblies and whose edges connect those assemblies that are adjacent.

\begin{algorithm}
\caption{\textsc{Is-Producible-Assembly-Fast}$(\alpha,\tau)$}
\label{alg-fast}
\begin{algorithmic}[1]
  \STATE {\bf input:} assembly $\alpha$ and temperature $\tau$
  \STATE $V_c \leftarrow \setr{ \{v\} }{ v \in \dom \alpha }$\ \ \ // \emph{(positions defining) subassemblies of $\alpha$; initially individual (positions of) tiles}
  \STATE $E_c \leftarrow \{\{\{u\},\{v\}\} \ |\ \text{$\{u\} \in V_c$ and $\{v\} \in V_c$ and $u$ and $v$ are adjacent and interact} \}$
  \STATE $H \leftarrow $ empty heap
  \FORALL{$\{\{u\},\{v\}\} \in E_c$} \label{outer-for}
    \STATE $w(\{u\},\{v\}) \leftarrow $ strength of glue binding $\alpha(u)$ and $\alpha(v)$
    \STATE insert($\{\{u\},\{v\}\}, H)$
  \ENDFOR
  \WHILE{$|V_c| > 1$} \label{outer-while}
    \STATE $\{C_1,C_2\} \leftarrow $ remove-max($H$)\ \ \ // \emph{assume $|C_1| \geq |C_2|$ without loss of generality}
    \label{repeat-body}
    \IF{$w(C_1,C_2) < \tau$}
      \PRINT{``$\alpha$ is not producible''}  and {\bf exit}
    \ENDIF
    \STATE remove $C_2$ from $V_c$
    \FORALL{neighbors $C$ of $C_2$} \label{inner-for}
      \STATE remove $\{C_2,C\}$ from $E_c$ and $H$ \label{remove}
      \IF{$\{C_1,C\} \in E_c$} \label{check-edge}
        \STATE $w(C_1,C) \leftarrow w(C_1,C) + w(C_2,C)$
        \STATE increase-key($\{C_1,C\},H)$\ \ \ // \emph{update edge $\{C_1,C\}$ with new weight}
        \label{increase-key}
      \ELSE
        \STATE $w(C_1,C) \leftarrow w(C_2,C)$
        \STATE add $\{C_1,C\}$ to $E_c$ and $H$
      \ENDIF
    \ENDFOR
  \ENDWHILE
  \PRINT{``$\alpha$ is producible''}
\end{algorithmic}
\end{algorithm}

Each edge of $G_c$ stores the total glue strength between the assemblies.
Whenever two assemblies $C_1$ and $C_2$, with $|C_1| \geq |C_2|$ without loss of generality, are combined to form a new assembly, $G_c$ is updated by removing $C_2$, merging its edges with those of $C_1$, and for any edges they already share (i.e., the neighbor on the other end of the edge is the same), summing the strengths on the edges.
Each update of an edge (adding it to $C_1$, or finding it in $C_1$ to update its strength) can be done in $O(\log n)$ time using a tree set data structure to store neighbors for each assembly.

We claim that the total number of such updates of all edges is $O(n \log n)$ over all time, or amortized $O(\log n)$ updates per iteration of the outer loop.
To see why, observe that the number of edges an assembly has is at most linear in its size, so the number of new edges that must be added to $C_1$, or existing edges in $C_1$ whose strengths must be updated, is at most (within a constant) the size of the smaller component $C_2$.
The total number of edge updates is then, if $\Upsilon$ is the assembly tree discovered by the algorithm, $\sum_{\text{nodes } u \in \Upsilon} \min\{|\text{left}(u)|,|\text{right}(u)|\}$, where $|\text{left}(u)|$ and $|\text{right}(u)|$ respectively refer to the number of leaves of $u$'s left and right subtrees.
For a given number $n$ of leaves, this sum is maximized with a balanced tree, and in that case (summing over all levels of the tree) is $\sum_{i=0}^{\log n} 2^i (n / 2^i) = O(n \log n)$.
So the total time to update all edges is $O(n \log^2 n)$.

As for actually finding $C_1$ and $C_2$, each iteration of the outer loop, we can just look for the pair of adjacent assemblies with the largest connection strength.
So store the edges in a heap and we can simply grab the strongest edge off the top and update the heap by updating the keys containing $C_1$ whose connection strength changed and removing those containing $C_2$ but not $C_1$.
The edges whose connection strength changed correspond to precisely those neighbors that $C_1$ and $C_2$ shared before being merged.
Therefore $|C_2|$ is an upper bound on the number of edge updates required.
Thus the amortized number of heap updates is $O(\log n)$ per iteration of the outer loop by the same argument as above.
Thus it takes amortized time $O(\log^2 n)$ per iteration if each heap operation is $O(\log n)$.
Therefore (this non-naive implementation of) the algorithm takes $O(n \log^2 n)$ time.

The algorithm \textsc{Is-Producible-Assembly-Fast} (Algorithm~\ref{alg-fast}) implements this optimized idea.
The terminology for heap operations is taken from~\cite{CLRS01}.
Note that the way we remove $C_1$ and $C_2$ and add their union is to simply delete $C_2$ and then update $C_1$ to contain $C_2$'s edges.
The graph $G_c$ discussed above is $G_c = (V_c,E_c)$ where $V_c$ and $E_c$ are variables in \textsc{Is-Producible-Assembly-Fast}.
The weight function $w$ is used by the heap $H$ to order its elements.

Summarizing the analysis, each data structure operation takes time $O(\log n)$ with appropriate choice of a backing data structure.
The two outer loops (lines~\ref{outer-for} and~\ref{outer-while}) take $O(n)$ iterations.
The inner loop (line~\ref{inner-for}) runs for amortized $O(\log n)$ iterations, and its body executes a constant number of $O(\log n)$ time operations.
Therefore the total running time is $O(n \log^2 n)$.

\end{proof}

\vspace{-0.2in}
\section{Linear-time verification of temperature 1 unique production}
\label{sec-upv-t1}

This section shows that there is an algorithm, faster than the previous known algorithm~\cite{TwoHandsBetterThanOne}, that solves the temperature 1 \emph{unique producibility verification} (UPV) problem: given an assembly $\alpha$ and a temperature-1 hierarchical tile system $\calT$, decide if $\alpha$ is the unique producible, terminal assembly of $\calT$.
This is done by showing an algorithm for the temperature 1 UPV problem in the seeded model (which is faster than the general-temperature algorithm of~\cite{ACGHKMR02}), and then applying the technique of~\cite{TwoHandsBetterThanOne} relating producibility and terminality in the temperature 1 seeded and hierarchical models.


Define the decision problems $\sUPV$ and $\hUPV$ by the language $\setr{(\calT,\alpha)}{\termasm{\calT} = \{\alpha\}},$ where $\calT$ is a temperature 1 seeded TAS in the former case and a temperature 1 hierarchical TAS in the latter case.
To simplify the time analysis we assume $|\calT| = O(|\alpha|)$.


The following is the only result in this paper on the seeded aTAM.

\begin{thm}\label{thm-alg-temp1-seeded}
  There is an algorithm that solves the $\sUPV$ problem in time $O(|\alpha| \log |\calT|)$.
\end{thm}

\begin{proof}
    Let $\calT=(T,s,1)$ and $\alpha$ be a instance of the $\sUPV$ problem.
    We first check that every tile in $\alpha$ appears in $T$, which can be done in time $O(|\alpha| \log |T|)$ by storing elements of $T$ in a data structure supporting $O(\log n)$ time access.
    In the seeded aTAM at temperature 1, $\alpha$ is producible if and only if it contains the seed $s$ and its binding graph is connected, which can be checked in time $O(|\alpha|)$.
    We must also verify that $\alpha$ is terminal, which is true if and only if all glues on unbound sides are null, checkable in time $O(|\alpha|)$.

    Once we have verified that $\alpha$ is producible and terminal, it remains to verify that $\calT$ uniquely produces $\alpha$. 
    Adleman, Cheng, Goel, Huang, Kempe, Moisset de Espan\'{e}s, and Rothemund~\cite{ACGHKMR02} showed that this is true (at any temperature) if and only if, for every position $p \in \dom\alpha$, if $\alpha_p \sqsubset \alpha$ is the maximal producible subassembly of $\alpha$ such that $p \not \in \dom \alpha_p$, then $\alpha(p)$ is the only tile type attachable to $\alpha_p$ at position $p$.
    They solve the problem by producing each such $\alpha_p$ and checking whether there is more than one tile type attachable to $\alpha_p$ at $p$.
    We use a similar approach, but we avoid producing each $\alpha_p$ by exploiting special properties of temperature 1 producibility.

    Given $p,q\in\dom\alpha$ such that $p \neq q$, write $p \prec q$ if, for every producible assembly $\beta$, $q \in \dom \beta \implies p \in \dom \beta$, i.e., the tile at position $p$ must be present before the tile at position $q$ can be attached.
    We must check each $p \in\dom \alpha$ and each position $q\in\dom\alpha$ adjacent to $p$ such that $p \not \prec q$ to see whether a tile type $t \neq \alpha(p)$ shares a positive-strength glue with $\alpha(q)$ in direction $q-p$ (i.e., whether, if $\alpha(p)$ were not present, $t$ could attach at $p$ instead).
    If we know which positions $q$ adjacent to $p$ satisfy $p \not\prec q$, this check can be done in time $O(\log |T|)$ with appropriate choice of data structure, implying total time $O(|\alpha| \log |T|)$ over all positions $p \in \dom \alpha$.
    It remains to show how to determine which adjacent positions $p,q \in \dom \alpha$ satisfy $p \prec q$.

    Recall that a \emph{cut vertex} of a connected graph is a vertex whose removal disconnects the graph, and that a subgraph is \emph{biconnected} if the removal of any single vertex from the subgraph leaves it connected.
    Every graph can be decomposed into a tree of biconnected components, with cut vertices connecting different biconnected components (and belonging to all biconnected components that they connect).
    If $p$ is not a cut vertex of the binding graph of $\alpha$, then $\dom \alpha_p$ is simply $\dom \alpha \setminus \{p\}$ because, for all $q \in \dom \alpha \setminus \{p\}$, $p \not\prec q$.
    If $p$ is a cut vertex, then $p \prec q$ if and only if 
    removing $p$ from the binding graph of $\alpha$ places $q$ and the seed position in two different connected components.
    This is because the connected component containing the seed after removing $p$ corresponds precisely to $\alpha_p$.

    Run the linear time Hopcroft-Tarjan algorithm~\cite{HopcroftTarjanBiconnected73} for decomposing the binding graph of $\alpha$ into a tree of its biconnected components, which also identifies which vertices in the graph are cut vertices and which biconnected components they connect.
    Recall that the Hopcroft-Tarjan algorithm is an augmented depth-first search.
    Root the tree with $s$'s biconnected component (i.e., start the depth-first search there), so that each component has a parent component and child components.
    In particular, each cut vertex $p$ has a ``parent'' biconnected component and $k \geq 1$ ``child'' biconnected components.
    Removing such a cut vertex $p$ will separate the graph into $k+1$ connected components: the nodes in the $k$ subtrees and the remaining nodes connected to the parent biconnected component of $p$.
    Thus $p \prec q$ if and only if $p$ is a cut vertex and $q$ is contained in the subtree rooted at $p$.


    This check can be done for all positions $p$ and their $\leq 4$ adjacent positions $q$ in linear time by ``weaving'' the checks into the Hopcroft-Tarjan algorithm.
    As the depth-first search executes, each vertex $p$ is marked as either \emph{unvisited}, \emph{visiting} (meaning the search is currently in a subtree rooted at $p$), or \emph{visited} (meaning the search has visited and exited the subtree rooted at $p$).
    If $p$ is marked as visited or unvisited at the time $q$ is processed, then $q$ is not in the subtree under $p$. 
    If $p$ is marked as visiting when $q$ is processed, then $q$ is in $p$'s subtree. 

    At the time $q$ is visited, it may not yet be known whether $p$ is a cut vertex.
    To account for this, simply run the Hopcroft-Tarjan algorithm twice, doing the checks just described on the second execution, using the cut vertex information obtained on the first execution.
\end{proof}

\begin{thm}\label{thm-alg-temp1-hier}
  There is an algorithm that solves the $\hUPV$ problem in time $O(|\alpha| |\calT| \log |\calT|)$.
\end{thm}

\begin{proof}
  Cannon, Demaine, Demaine, Eisenstat, Patitz, Schweller, Summers, and Winslow~\cite{TwoHandsBetterThanOne} showed that a temperature 1 hierarchical TAS $\calT=(T,1)$ uniquely produces $\alpha$ if and only if, for each $s \in T$, the seeded TAS $\calT_s = (T,s,1)$ uniquely produces $\alpha$.
  Therefore, the $\hUPV$ problem can be solved by calling the algorithm of Theorem~\ref{thm-alg-temp1-seeded} $|\calT|$ times, resulting in a running time of $O(|\alpha| |\calT| \log |\calT|)$.
\end{proof}

\vspace{-0.2in}
\section{Unions of producible assemblies are producible}
\label{sec-union}

Throughout this section, fix a hierarchical TAS $\calT=(T,\tau)$.
Let $\alpha,\beta$ be assemblies.
We say $\alpha$ and $\beta$ are \emph{consistent} if $\alpha(p) = \beta(p)$ for all points $p \in \dom \alpha \cap \dom \beta$.
If $\alpha$ and $\beta$ are consistent, let $\alpha \cup \beta$ be defined as the assembly $(\alpha\cup\beta)(p) = \alpha(p)$ if $\alpha$ is defined, and $(\alpha\cup\beta)(p)=\beta(p)$ if $\alpha(p)$ is undefined.
If $\alpha$ and $\beta$ are not consistent, let $\alpha \cup \beta$ be undefined.

\begin{thm} \label{thm-union-producible}
  If $\alpha,\beta$ are producible assemblies that are consistent and $\dom\alpha \cap \dom\beta \neq \emptyset$, then $\alpha \cup \beta$ is producible.
  Furthermore, $\alpha \to \alpha \cup \beta$, i.e., it is possible to assemble exactly $\alpha$, then to assemble the missing portions of $\beta$.
\end{thm}

\begin{proof}
  If $\alpha$ and $\beta$ are consistent and have non-empty overlap, then $\alpha \cup \beta$ is necessarily stable, since every cut of $\alpha \cup \beta$ is a superset of some cut of either $\alpha$ or $\beta$, which are themselves stable.

    \begin{figure}[htb]
    \centering
      \subfigure
        [{\figuresize First operation to combine the assembly trees for $\alpha$ and $\beta$.  $l_1$ and $l_2$ are two leaves representing the same position in $\dom \alpha \cap \dom \beta$.}]{
        \includegraphics[width=3in]{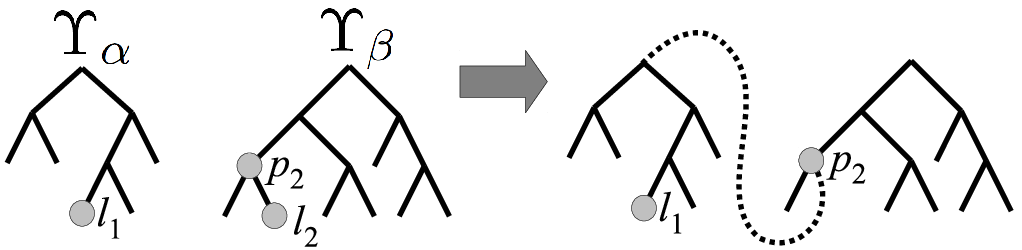}
        \label{fig:tree-surgery-first-step}
      }
      \hfill
      \subfigure
        [{\figuresize Operation to eliminate one of two leaves $l_1$ and $l_2$ representing the same tile in the tree while preserving that all attachments are stable.}]{
        \includegraphics[width=3in]{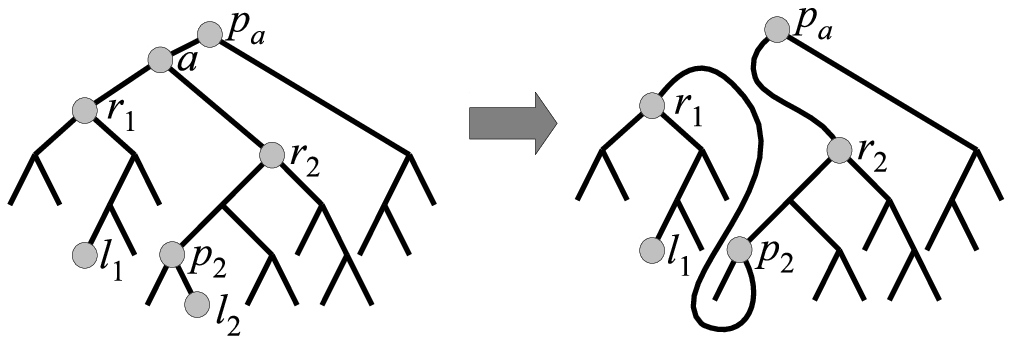}
        \label{fig:tree-surgery}
      }
      \label{fig:tree-surgery-whole}
      \caption{Procedure to construct assembly tree for $\alpha \cup \beta$ from separate assembly trees for $\alpha$ and $\beta$.}
    \end{figure}

  Let $\Ua$ and $\Ub$ be assembly trees for $\alpha$ and $\beta$, respectively.
  Define an assembly tree $\U$ for $\alpha \cup \beta$ by the following construction.
  Let $l_1$ be a leaf in $\Ua$ and let $l_2$ be a leaf in $\Ub$ representing the same position $x \in \dom\alpha \cap \dom\beta$, as shown in Figure~\ref{fig:tree-surgery-first-step}.
  Remove $l_2$ and replace it with the entire tree $\Ua$.
  Call the resulting tree $\U'$.
  At this point, $\U'$ is not an assembly tree if $\alpha$ and $\beta$ overlapped on more than one point, because every position in $\dom\alpha \cap \dom\beta \setminus \{x\}$ has duplicated leaves.
  Therefore the tree $\U'$ is not a hierarchical division of the set $\dom \alpha \cup \dom \beta$, since not all unions represented by each internal node are disjoint unions.
  However, each such union does represent a stable assembly.
  We will show how to modify $\U'$ to eliminate each of these duplicates, while maintaining the invariant that each internal node represents a stable attachment, and we call the resulting assembly tree $\U$.
  Furthermore, the subtree $\Ua$ that was placed under $p_2$ will not change as a result of these modifications, which implies $\alpha \to \alpha \cup \beta$.


  The process to eliminate one pair of duplicate leaves is shown in Figure \ref{fig:tree-surgery}.
  Let $l_1$ and $l_2$ be two leaves representing the same point in $\dom\alpha \cap \dom\beta$, and let $a$ be their least common ancestor in $\U$, noting that $a$ is not contained in $\Ua$ since $l_2$ is not contained in $\Ua$.
  Let $p_a$ be the parent of $a$.
  Let $r_1$ be the root of the subtree under $a$ containing $l_1$.
  Let $r_2$ be the root of the subtree under $a$ containing $l_2$.
  Let $p_2$ be the parent of $l_2$.
  Remove the leaf $l_2$ and the node $a$.
  Set the parent of $r_1$ to be $p_2$.
  Set the parent of $r_2$ to be $p_a$.

  Since we have replaced the leaf $l_2$ with a subtree containing the leaf $l_1$, the subtree rooted at $r_1$ is an assembly containing the tile represented by $l_2$, in the same position.
  Since the original attachment of $l_2$ to its sibling was stable, by monotonicity, the attachment represented by $p_2$ is still legal.
  The removal of $a$ is simply to maintain that $\Upsilon$ is a full binary tree; leaving it would mean that it represents a superfluous ``attachment'' of the assembly $r_2$ to $\emptyset$.
  However, it is now legal for $r_2$ to be a direct child of $p_a$, since $r_2$ (due to the insertion of the entire $r_1$ subtree beneath a descendant of $r_2$, again by monotonicity) now has all the tiles necessary for its attachment to the old sibling of $a$ to be stable.
  Since $a$ was not contained in $\Ua$, the subtree $\Ua$ has not been altered.

  This process is iterated for all duplicate leaves.
  When all duplicates have been removed, $\Upsilon$ is a valid assembly tree with root $\alpha \cup \beta$. 
  Since $\Upsilon$ contains $\Ua$ as a subtree, $\alpha \to \alpha \cup \beta$.
\end{proof}

\vspace{-0.2in}
\section{Conclusion}
\label{sec-conclusion}

Theorem~\ref{thm-union-producible} shows that if assemblies $\alpha$ and $\beta$ overlap consistently, then $\alpha \cup \beta$ is producible.
What if $\alpha=\beta$?
Suppose we have three copies of $\alpha$, and label them each uniquely as $\alpha_1,\alpha_2,\alpha_3$.
Suppose further than $\alpha_2$ overlaps consistently with $\alpha_1$ when translated by some non-zero vector $\vv$.
Then we know that $\alpha_1 \cup \alpha_2$ is producible.
Suppose that $\alpha_3$ is $\alpha_2$ translated by $\vv$, or equivalently it is $\alpha_1$ translated by $2 \vv$.
Then $\alpha_2 \cup \alpha_3$ is producible, since this is merely a translated copy of $\alpha_1 \cup \alpha_2$.
It seems intuitively that $\alpha_1 \cup \alpha_2 \cup \alpha_3$ should be producible as well.
However, while $\alpha_1$ overlaps consistently with $\alpha_2$, and $\alpha_2$ overlaps consistently with $\alpha_3$, it could be the case that $\alpha_3$ intersects $\alpha_1$ inconsistently, i.e., they share a position but put a different tile type at that position.
In this case $\alpha_1 \cup \alpha_2 \cup \alpha_3$ is undefined.
See Figure~\ref{fig:pump} for an example.

\begin{figure}[htb]
\begin{center}
  \includegraphics[width=6in]{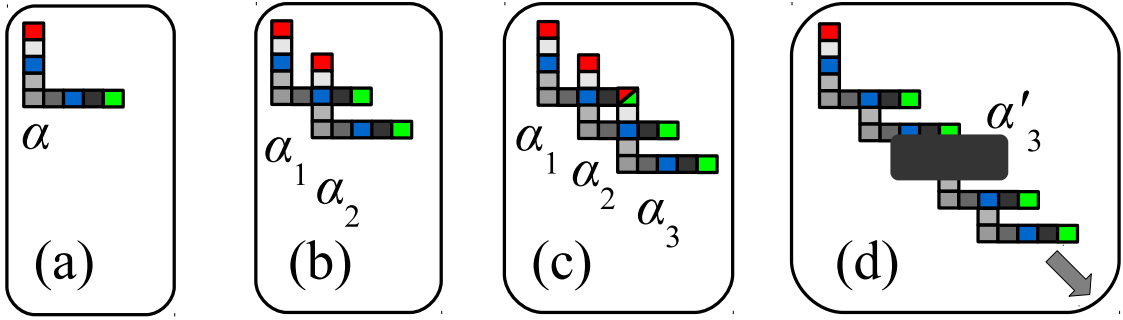}
  \caption{\label{fig:pump} \figuresize
  (a) A producible assembly $\alpha$. Gray tiles are all distinct types from each other, but red, green, and blue each represent one of three different tile types, so the two blue tiles are the same type.
  (b) By Theorem~\ref{thm-union-producible}, $\alpha_1 \cup \alpha_2$ is producible, where $\alpha_1 = \alpha$ and $\alpha_2 = \alpha_1 + (2,-2)$, because they overlap in only one position, and they both have the blue tile type there.
  (c) $\alpha_1$ and $\alpha_3$ both have a tile at the same position, but the types are different (red in the case of $\alpha_1$ and green in the case of $\alpha_3$).
  (d) However, a subassembly $\alpha'_i$ of each new $\alpha_i$ can grow, enough to allow the translated equivalent subassembly $\alpha'_{i+1}$ of $\alpha_{i+1}$ to grow from $\alpha'_i$, so an infinite structure is producible.
  }
\end{center}
\end{figure}

However, in this case, although $\alpha_1 \cup \alpha_2 \cup \alpha_3$ is not producible (in fact, not even defined), ``enough'' of $\alpha_3$ (say, $\alpha'_3 \sqsubset \alpha_3$) can grow off of $\alpha_1 \cup \alpha_2$ to allow a fourth copy $\alpha'_4$ to begin to grow to an assembly to which a fifth copy $\alpha'_5$ can attach, etc., so that an infinite assembly can grow by ``pumping'' additional copies of $\alpha'_3$.
Is this always possible?
In other words, is it the case that if $\alpha$ is a producible assembly of a hierarchical TAS $\calT$, and $\alpha$ overlaps consistently with some non-zero translation of itself, then $\calT$ necessarily produces an infinite assembly?
If true, this would imply that no hierarchical TAS producing such an assembly could be uniquely produce a finite shape.
This would settle an open question posed by Chen and Doty~\cite{CheDot12}, who showed that as long as a hierarchical TAS does not produce assemblies that consistently overlap any translation of themselves, then the TAS cannot uniquely produce any shape in time sublinear in its diameter.

\paragraph{Acknowledgements.}
The author is very grateful to Ho-Lin Chen, David Soloveichik, Damien Woods, Matt Patitz, Scott Summers, Robbie Schweller, J\'{a}n Ma{\v{n}}uch, and Ladislav Stacho for many insightful discussions.

\newpage
\bibliographystyle{plain}
\bibliography{tam}
\newpage

\newpage
\appendix
\section{Formal definition of the abstract tile assembly model}
\label{sec-tam-formal}

\newcommand{\fullgridgraph}{G^\mathrm{f}}
\newcommand{\bindinggraph}{G^\mathrm{b}}

This section gives a terse definition of the abstract Tile Assembly Model (aTAM,~\cite{Winf98}). This is not a tutorial; for readers unfamiliar with the aTAM,  \cite{RotWin00} gives an excellent introduction to the model.

Fix an alphabet $\Sigma$. $\Sigma^*$ is the set of finite strings over $\Sigma$.
Given a discrete object $O$, $\langle O \rangle$ denotes a standard encoding of $O$ as an element of $\Sigma^*$.
$\Z$, $\Z^+$, and $\N$ denote the set of integers, positive integers, and nonnegative integers, respectively.
For a set $A$, $\calP(A)$ denotes the power set of $A$.
Given $A \subseteq \Z^2$, the \emph{full grid graph} of $A$ is the undirected graph $\fullgridgraph_A=(V,E)$, where $V=A$, and for all $u,v\in V$, $\{u,v\} \in E \iff \| u-v\|_2 = 1$; i.e., if and only if $u$ and $v$ are adjacent on the integer Cartesian plane.
A \emph{shape} is a set $S \subseteq \Z^2$ such that $\fullgridgraph_S$ is connected.

A \emph{tile type} is a tuple $t \in (\Sigma^* \times \N)^4$; i.e., a unit square with four sides listed in some standardized order, each side having a \emph{glue label} (a.k.a. \emph{glue}) $\ell \in \Sigma^*$ and a nonnegative integer \emph{strength}, denoted $str(\ell)$.
For a set of tile types $T$, let $\Lambda(T) \subset \Sigma^*$ denote the set of all glue labels of tile types in $T$.
If a glue has strength 0, we say it is \emph{null}, and if a positive-strength glue facing some direction does not appear on some tile type in the opposite direction, we say it is \emph{functionally null}.
We assume that all tile sets in this paper contain no functionally null glues.\footnote{This assumption does not affect the results of this paper.
    It is irrelevant for Theorem~\ref{thm-union-producible} or the correctness of the algorithms in the other theorems.
    It also does not affect the running time results for algorithms taking a TAS as input, because we can preprocess $T$ in linear time to find and set to null any functionally null glues.
    The number of glues is $O(|T|)$, and we assume that each glue from glue set $G$ is an integer in the set $\{0,\ldots,|G|-1\}$.
    We can use a Boolean array of size $|G|$ to determine in time $O(|T|)$ which glues appear on the north that do not appear on the south of some tile type.
    Repeat this for each of the remaining three directions.
    Then replace all functionally null glues in $T$ with null glues, which takes time $O(|T|)$.
    To do this replacement in an assembly $\alpha$ takes time $O(|\alpha|)$.}
Let $\dall$ denote the \emph{directions} consisting of unit vectors $\{(0,1), (0,-1), (1,0), (-1,0)\}$.
Given a tile type $t$ and a direction $d \in \dall$, $t(d) \in \Lambda(T)$ denotes the glue label on $t$ in direction $d$.
We assume a finite set $T$ of tile types, but an infinite number of copies of each tile type, each copy referred to as a \emph{tile}. An \emph{assembly}
is a nonempty connected arrangement of tiles on the integer lattice $\Z^2$, i.e., a partial function $\alpha:\Z^2 \dashrightarrow T$ such that $\fullgridgraph_{\dom \alpha}$ is connected and $\dom \alpha \neq \emptyset$.
The \emph{shape of $\alpha$} is $\dom \alpha$.
Write $|\alpha|$ to denote $|\dom\alpha|$.
Given two assemblies $\alpha,\beta:\Z^2 \dashrightarrow T$, we say $\alpha$ is a \emph{subassembly} of $\beta$, and we write $\alpha \sqsubseteq \beta$, if $\dom \alpha \subseteq \dom \beta$ and, for all points $p \in \dom \alpha$, $\alpha(p) = \beta(p)$.

Given two assemblies $\alpha$ and $\beta$, we say $\alpha$ and $\beta$ are \emph{equivalent up to translation}, written $\alpha \simeq \beta$, if there is a vector $\vec{x} \in \Z^2$ such that $\dom\alpha = \dom\beta + \vec{x}$ (where for $A \subseteq \Z^2$, $A + \vec{x}$ is defined to be $\setr{p + \vec{x}}{p \in A}$) and for all $p \in \dom\beta$, $\alpha(p + \vec{x}) = \beta(p)$.
In this case we say that $\beta$ is a \emph{translation} of $\alpha$.
We have fixed assemblies at certain positions on $\Z^2$ only for mathematical convenience in some contexts, but of course real assemblies float freely in solution and do not have a fixed position.

Let $\alpha$ be an assembly and let $p\in\dom\alpha$ and $d\in\dall$ such that $p + d \in \dom\alpha$.
Let $t=\alpha(p)$ and $t' = \alpha(p+d)$.
We say that the tiles $t$ and $t'$ at positions $p$ and $p+d$ \emph{interact} if $t(d) = t'(-d)$ and $str(t(d)) > 0$, i.e., if the glue labels on their abutting sides are equal and have positive strength.
Each assembly $\alpha$ induces a \emph{binding graph} $\bindinggraph_\alpha$, a grid graph $G=(V_\alpha,E_\alpha)$, where $V_\alpha=\dom\alpha$, and $\{p_1,p_2\} \in E_\alpha \iff \alpha(p_1) \text{ interacts with } \alpha(p_2)$.\footnote{For $\fullgridgraph_{\dom \alpha}=(V_{\dom \alpha},E_{\dom \alpha})$ and $\bindinggraph_\alpha=(V_\alpha,E_\alpha)$, $\bindinggraph_\alpha$ is a spanning subgraph of $\fullgridgraph_{\dom \alpha}$: $V_\alpha = V_{\dom \alpha}$ and $E_\alpha \subseteq E_{\dom \alpha}$.}
Given $\tau\in\Z^+$, $\alpha$ is \emph{$\tau$-stable} if every cut of $\bindinggraph_\alpha$ has weight at least $\tau$, where the weight of an edge is the strength of the glue it represents.
That is, $\alpha$ is $\tau$-stable if at least energy $\tau$ is required to separate $\alpha$ into two parts.
When $\tau$ is clear from context, we say $\alpha$ is \emph{stable}.

\opt{submission}{\paragraph{Seeded aTAM.}}
\opt{normal}{\subsection{Seeded aTAM}}
A \emph{seeded tile assembly system} (seeded TAS) is a triple $\calT = (T,\sigma,\tau)$, where $T$ is a finite set of tile types, $\sigma:\Z^2 \dashrightarrow T$ is the finite, $\tau$-stable \emph{seed assembly},
and $\tau\in\Z^+$ is the \emph{temperature}.
Let $|\calT|$ denote $|T|$.
If $\calT$ has a single seed tile $s\in T$ (i.e., $\sigma(0,0)=s$ for some $s \in T$ and is undefined elsewhere), then we write $\calT=(T,s,\tau).$
Given two $\tau$-stable assemblies $\alpha,\beta:\Z^2 \dashrightarrow T$, we write $\alpha \to_1^{\calT} \beta$ if $\alpha \sqsubseteq \beta$ and $|\dom \beta \setminus \dom \alpha| = 1$. In this case we say $\alpha$ \emph{$\calT$-produces $\beta$ in one step}.\footnote{Intuitively $\alpha \to_1^\calT \beta$ means that $\alpha$ can grow into $\beta$ by the addition of a single tile; the fact that we require both $\alpha$ and $\beta$ to be $\tau$-stable implies in particular that the new tile is able to bind to $\alpha$ with strength at least $\tau$. It is easy to check that had we instead required only $\alpha$ to be $\tau$-stable, and required that the cut of $\beta$ separating $\alpha$ from the new tile has strength at least $\tau$, then this implies that $\beta$ is also $\tau$-stable.}
If $\alpha \to_1^{\calT} \beta$, $ \dom \beta \setminus \dom \alpha=\{p\}$, and $t=\beta(p)$, we write $\beta = \alpha + (p \mapsto t)$.

A sequence of $k\in\Z^+$ 
assemblies $\vec{\alpha} = (\alpha_0,\alpha_1,\ldots,\alpha_{k-1})$ is a \emph{$\calT$-assembly sequence} if, for all $1 \leq i < k$, $\alpha_{i-1} \to_1^\calT \alpha_{i}$.
We write $\alpha \to^\calT \beta$, and we say $\alpha$ \emph{$\calT$-produces} $\beta$ (in 0 or more steps) if there is a $\calT$-assembly sequence $\vec{\alpha}=(\alpha,\alpha_1,\alpha_2,\ldots,\alpha_{k-1} = \beta)$ of length $k = |\dom \beta \setminus \dom \alpha| + 1$. 
We say $\alpha$ is \emph{$\calT$-producible} if $\sigma \to^\calT \alpha$, and we write $\prodasm{\calT}$ to denote the set of $\calT$-producible assemblies.
The relation $\to^\calT$ is a partial order on $\prodasm{\calT}$ \cite{Roth01,jSSADST}. 

An assembly $\alpha$ is \emph{$\calT$-terminal} if $\alpha$ is $\tau$-stable and $\partial^\calT \alpha=\emptyset$.
We write $\termasm{\calT} \subseteq \prodasm{\calT}$ to denote the set of $\calT$-producible, $\calT$-terminal assemblies.

A seeded TAS $\calT$ is \emph{directed (a.k.a., deterministic, confluent)} if the poset $(\prodasm{\calT}, \to^\calT)$ is directed; i.e., if for each $\alpha,\beta \in \prodasm{\calT}$, there exists $\gamma\in\prodasm{\calT}$ such that $\alpha \to^\calT \gamma$ and $\beta \to^\calT \gamma$.\footnote{The following two convenient characterizations of ``directed'' are routine to verify.
$\calT$ is directed if and only if $|\termasm{\calT}| = 1$.
$\calT$ is \emph{not} directed if and only if there exist $\alpha,\beta\in\prodasm{\calT}$ and $p \in \dom \alpha \cap \dom \beta$ such that $\alpha(p) \neq \beta(p)$.}
We say that $\calT$ \emph{uniquely produces} $\alpha$ if $\termasm{\calT} = \{\alpha\}$.

\opt{submission}{\paragraph{Hierarchical aTAM.}}
\opt{normal}{\subsection{Hierarchical aTAM}}
A \emph{hierarchical tile assembly system} (hierarchical TAS) is a pair $\calT = (T,\tau)$, where $T$ is a finite set of tile types,
and $\tau\in\Z^+$ is the \emph{temperature}.
Let $\alpha,\beta:\Z^2 \dashrightarrow T$ be two assemblies.
Say that $\alpha$ and $\beta$ are \emph{nonoverlapping} if $\dom\alpha \cap \dom\beta = \emptyset$.
If $\alpha$ and $\beta$ are nonoverlapping assemblies, define $\alpha \cup \beta$ to be the assembly $\gamma$ defined by $\gamma(p) = \alpha(p)$ for all $p\in\dom\alpha$, $\gamma(p) = \beta(p)$ for all $p \in \dom \beta$, and $\gamma(p)$ is undefined for all $p \in \Z^2 \setminus (\dom\alpha \cup \dom\beta)$.
An assembly $\gamma$ is \emph{singular} if $\gamma(p)=t$ for some $p\in\Z^2$ and some $t \in T$ and $\gamma(p')$ is undefined for all $p' \in \Z^2 \setminus \{p\}$.
Given a hierarchical TAS $\calT=(T,\tau)$, an assembly $\gamma$ is \emph{$\calT$-producible} if either 1) $\gamma$ is singular, or 2) there exist producible nonoverlapping assemblies $\alpha$ and $\beta$ such that $\gamma = \alpha\cup\beta$ and $\gamma$ is $\tau$-stable. 
In the latter case, write $\alpha + \beta \to_1^\calT \gamma$.
An assembly $\alpha$ is \emph{$\calT$-terminal} if for every producible assembly $\beta$ such that $\alpha$ and $\beta$ are nonoverlapping, $\alpha \cup \beta$ is not $\tau$-stable.\footnote{The restriction on overlap is a model of a chemical phenomenon known as \emph{steric hindrance}~\cite[Section 5.11]{WadeOrganicChemistry91} or, particularly when employed as a design tool for intentional prevention of unwanted binding in synthesized molecules, \emph{steric protection}~\cite{HellerPugh1,HellerPugh2,GotEtAl00}.}
Define $\prodasm{\calT}$ to be the set of all $\calT$-producible assemblies.
Define $\termasm{\calT} \subseteq \prodasm{\calT}$ to be the set of all $\calT$-producible, $\calT$-terminal assemblies.
A hierarchical TAS $\calT$ is \emph{directed (a.k.a., deterministic, confluent)} if $|\termasm{\calT}|=1$. 
We say that $\calT$ \emph{uniquely produces} $\alpha$ if $\termasm{\calT} = \{\alpha\}$.

Let $\calT$ be a hierarchical TAS, and let $\widehat{\alpha} \in \prodasm{\calT}$ be a $\calT$-producible assembly.
An \emph{assembly tree} $\Upsilon$ of $\widehat{\alpha}$ is a full binary tree with $|\widehat{\alpha}|$ leaves, whose nodes are labeled by $\calT$-producible assemblies, with $\widehat{\alpha}$ labeling the root, singular assemblies labeling the leaves, and node $u$ labeled with $\gamma$ having children $u_1$ labeled with $\alpha$ and $u_2$ labeled with $\beta$, with the requirement that $\alpha + \beta \to_1^\calT \gamma$.
That is, $\Upsilon$ represents one possible pathway through which $\widehat{\alpha}$ could be produced from individual tile types in $\calT$.
Let $\Upsilon(\calT)$ denote the set of all assembly trees of $\calT$.
If $\alpha$ is a descendant node of $\beta$ in an assembly tree of $\calT$, write $\alpha \to^\calT \beta$.
Say that an assembly tree is \emph{$\calT$-terminal} if its root is a $\calT$-terminal assembly.
Let $\Upsilon_\Box(\calT)$ denote the set of all $\calT$-terminal assembly trees of $\calT$.
Note that even a directed hierarchical TAS can have multiple terminal assembly trees that all have the same root terminal assembly.

When $\calT$ is clear from context, we may omit $\calT$ from the notation above and instead write
$\to_1$,
$\to$,
$\partial \alpha$,
\emph{assembly sequence},
\emph{produces},
\emph{producible}, and
\emph{terminal}.


\end{document}